%% file: main.tex
\documentclass[letterpaper, 10 pt, conference]{ieeeconf}  

\usepackage{amsmath}
\usepackage{amssymb}
\usepackage{hyperref}
\usepackage{tikz}
\usepackage{pgfplots}
\usepackage{graphicx}

\IEEEoverridecommandlockouts                              

\title{\LARGE \bf
Global observability analysis of a growth model for insects farming
}

\author{Rania Tafat$^1$, Jaime A. Moreno $^2$, and Stefan Streif$^{1,3}$ 
\thanks{*This work was partially supported by UNAM-PAPIIT IN106323.}
\thanks{$^{1}$ R. Tafat and S. Streif are with the Technische Universit\"at Chemnitz, 09126 Chemnitz, Germany, Automatic Control and System Dynamics Lab; E-mail:
        {\tt\small \{rania.tafat, stefan.streif\}@etit.tu-chemnitz.de}}%
\thanks{$^{2}$ J. A. Moreno is with the Eléctrica y Computación
Instituto de Ingeniería-UNAM
Universidad Nacional Autónoma de México; E-mail:
        {\tt\small JMorenoP@iingen.unam.mx}}
\thanks{$^{3}$ S. Streif is also with the Fraunhofer Institute for Molecular Biology and Applied Ecology, Department of Bioresources, Giessen, Germany.}%
}

\usepackage{graphicx}
\usepackage{caption}
\usepackage{soul}
\usepackage{fancyhdr}
\usepackage{mwe}
\include{ACSD-preamble}

\fancypagestyle{FirstPage}{

\lfoot{\copyright 2024 the authors. This work has been accepted to ECC for publication under a Creative Commons Licence CC-BY-NC-ND}

}
\begin{document}

\maketitle 
\thispagestyle{empty}
\pagestyle{empty}
\thispagestyle{FirstPage}
\begin{abstract}
The \emph{Hermetia illucens} insects or the black soldier fly has been attracting a growing interest in the food and feed industry. For its high nutritional value on the one hand, and because it is an adequate species for insects in controlled environmental agriculture systems, on the other.  
Therefore, several models describing this larvae's behaviour have been developed in the literature. 
Due to the complex nature of living organisms, systems of controlled environment agriculture are characterised by their strong nonlinearities.
In this paper, we present a three dimensional nonlinear model describing the black soldier fly dry biomass weight dynamic changes due to the temperature's influence. 
In practice, this biomass weight is not measured in real time. 
This becomes problematic for applying feedback control strategies that assume full information of the states.
Thus, this work investigates the observability of the dry biomass of a \emph{Hermetia illucens} farming batch. 
The instantaneous and global observability of the aforementioned model is proven by constructing an injective transformation between the state space and a higher dimensional space where the transformed states are observable.
\end{abstract}


\section{Introduction}
{\emph{Hermetia illucens}, commonly known as black soldier fly, has been receiving increasing attention because of its usability for recycling organic waste and producing animal feed \cite{newton1977dried}.} 
It contains an important amount of protein \cite{barragan2017nutritional}, grows relatively fast \cite{nguyen2013influence} and reduces organic waste because it {grows on it} \cite{nguyen2015ability}.
This type of insect is currently used in livestock nutrition in aquaculture  \cite{belghit2019black}, \cite{bondari1981soldier}, and animal feed in general \cite{barragan2017nutritional}.
In addition, it is a promising candidate for human food because of its nutritional values \cite{wang2017review}. It has also been considered for human food supplements and production of bio plastic.

Insects, being {complex living organisms}, are affected by various factors during their growth process. 
Feed quality and rate, different temperatures of the environment, air quality, humidity etc. all have a significant impact on the insects' development  \cite{gligorescu2019development},
\cite{padmanabha2020larvaemodel},
\cite{tomberlin2009development} and \cite{varotto2017survey}. 
An effort in the {research community} has been done to obtain a better understanding of these different influences on the black soldier fly, resulting in many models \cite{chia2018threshold}, \cite{gligorescu2019development}, \cite{padmanabha2020larvaemodel}, \cite{PADMANABHA2023}, \cite{yuwatida2019}.
To our knowledge, the most extensive dynamic model of \emph{Hermetia illucens} is presented in \cite{padmanabha2020larvaemodel}, where the authors endeavour to include the majority of the relevant environmental factors impacting the insects growth.
The result is a highly complex dynamic discontinuous nonlinear system with 19 states and 12 inputs. 
The resulting model fits experimental data from various literature sources and explains the impact of each individual external factor on the development of the larvae during its life cycle.

These models are used to investigate the optimization of the industrial production process of the black soldier fly insect.
Model predictive control (MPC) was employed to identify the optimal growth trajectories of the insects during this industrial process.
This is particularly highlighted in \cite{PADMANABHA2023} where the authors show that the automation goals of the larvae production set by the optimal controller significantly reduce the production costs.
Similarly, in \cite{KHATIB2022} the authors use MPC to control a black soldier fly food unit with a lettuce one.
Another simulation of the same insects egg production shows the important cost reduction of the process while performing a higher production amount with MPC \cite{kobelski2022process}.

One common and important assumption in all the previous cited works regarding control, is the availability of full information on the biomass.  
However, in practice this is rarely possible.
Measuring the biomass weight, number or size is a difficult task and can only be done with low frequency, because it alters the whole growth and cultivation process. 
Such an invasive measuring process can only be conducted in relatively large time intervals, which is not enough for the application of optimal control schemes. 
An alternative to obtaining the information on the biomass without affecting its growth and with the adequate frequency, is the use of an observer or software sensor. The objective of this paper is to perform an observability analysis of the growth model, in order to decide if it is possible to design such an observer and the expected properties it can have.

Observability analysis, in its basic understanding, attempts to answer the following question "is it possible to \emph{uniquely} determine the states of a system given its outputs and inputs?"
A natural mathematical approach for answering this, is searching for the existence of an injective mapping between the states of the system and its outputs and inputs. 
Indeed, if such an injective map exists, this implies, by definition, that it is possible to uniquely determine the states from the available measurements.
For linear time invariant systems, this has been standardized in \cite{kalman1969topics}, where the transformation is unique and well known as the observability matrix.
This is a constant matrix that needs to be full column rank in order to be injective. 
Hence, the observability analysis procedure of linear systems can be summed up to checking the rank of the observability matrix.
For nonlinear systems however, such a standardization is very difficult to reach. 
The closest equivalent of the observability matrix in the general nonlinear case was presented in \cite{hermann1977nonlinear}. 
The candidate mapping is state and input dependent, it has high dimension and it is composed of the system's successive Lie derivatives of the output.
Although a general condition for its global injectivity is difficult to establish, Hermann et al. propose a general condition on the local injectivity of the transformation.
Indeed, if for some dimension, the Jacobian matrix of the proposed map is full column rank, then it is locally injective and consequently, the system is locally observable.
Since then, this method has been the standard way of performing observability analysis for nonlinear systems \cite{HellmanObsAnalysis}, \cite{martinelli2005observability, zhao2016observability, rausch2013nonlinear}.
Nevertheless, this approach has some limitations because it does not provide the size of the neighborhood for observability. 
And thus it is hard to determine what \emph{local} means when a nonlinear system is said to be locally observable. 
This could both mean an epsilon subset or the entire state space.
Moreover, strongly nonlinear systems can be locally observable for some neighborhood but globally unobservable. 
For these reasons, a global observability analysis is preferred. 
The analysis then becomes very attached to the model since, to the knowledge of the authors, there is no standard method. 
The approach consists of analysing the model, the outputs relations to the states of interests and to construct an injective map between the outputs/inputs and the states. 
This global observability analysis has been carried out on sensorless induction motors \cite{ibarra2004global}, uncertain reaction systems \cite{moreno2008global} and a class of biological reactors \cite{SCHAUM2007213}.
In this work, we aim at extending these examples by performing such an analysis on a simplified three dimensional temperature Hermetia Illucens model. The global observability of the system will be proven by constructing the immersion using differential observability tools.

The paper is organized as follows. Section \ref{Sec:Model} presents the \emph{Hermetia illucens} growth process modelling based on the temperatures of the environment. Section \ref{Sec:Observability} tackles the observability analysis of the previously mentioned model. The last Section \ref{Sec:Conclusion} concludes
the results.

\section{\emph{Hermetia illucens} Temperature model}
\label{Sec:Model} 
This section introduces a model describing the effect of the temperature on the growth process of the \emph{Hermetia illucens} insects in an indoor farming context.
This model will later be used to perfom an observability analysis.

The influences of interest on the larvae growth process are:
\begin{itemize}
\item $T_{\mathrm{med}}$, the temperature of the medium or substrate where the larvae grow, 
\item $T_{\mathrm{air}}$, the temperature of the air in the chamber where the experiment is set up,  and 
\item $T_{\mathrm{out}}$, the temperature of the outside environment of the chamber. 
\end{itemize}
The first two temperatures correspond to states of the model.  
Since they can be measured, they are also considered as system's outputs.
The outside temperature $T_{\mathrm{out}}$ is considered as a disturbance, since it cannot be influenced and it comes from the exterior of the system's boundaries. 
However, it is not necessarily an unwanted input that increases the system's error. 
And since it is a measured signal, it can be referred as a measured input signal or an external input signal \cite{van2010optimal}.

All the other environmental or external factors not previously mentioned but present in \cite{padmanabha2020larvaemodel} such as humidity, oxygen concentration of the air, moisture of the medium etc. are assumed to be held at a constant value and, therefore, they appear in the dynamical model as parameters. 

The life cycle of the black soldier fly is composed of four stages: egg, larval, pupae and adult stage \cite{li2011bioconversion}. 
The stage of interest in this work is the larval stage where the insects feed on organic material and grow rapidly. 
The metabolic process that describes the development of the larva during this stage is composed of assimilation, maintenance and growth. 
After which, the larvae starts maturing and its biomass declining.
However, in practice, the insects are harvested when their biomass weight $B_{\mathrm{dry}}$ reaches its peak, which only occurs at the end of the growth phase and before maturity.
Since the knowledge of the biomass weight is only needed before harvesting, the maturity phase is not of interest in this work and will not be taken into account.
Switching from a stage of the metabolic process to another is not only dependent on time but also on other growing conditions. 
In order to know precisely when this switching happens, the development sum was introduced in \cite{padmanabha2020larvaemodel}. 
In this work, the development sum is not considered. Instead and without loss of generality, the model considered that the switching between stages happens when the biomass reaches its peak. Which is a reasonable a priori information.

Chambers where larvae grows are equipped with ventilators and heaters/coolers. 
These are modelled as {(measured)} input control variables: $u_{\mathrm{v}}$ for the ventilator pump and $u_{\mathrm{T}}$ for the heater/cooler.

In the following, a temperature depending growth model of the \emph{Hermetia illucens} insect is introduced, which is a simplified version of the extensive model presented in  \cite{PADMANABHA2023}.

{Finally, the model of the closed production environment has the following state variables $x$, control inputs $u$, measured perturbation $d$ and measured outputs $y$}
\begin{equation}
    \begin{split}
        &x=\begin{bmatrix} B_{\mathrm{dry}} &  T_{\mathrm{med}} & T_{\mathrm{air}} \end{bmatrix}^\top,\\  &u=\begin{bmatrix} u_{\mathrm{v}} & u_{\mathrm{T}} \end{bmatrix}^{\top},\\  &d= T_{\mathrm{out}}, \\
        &y = \begin{bmatrix}
             T_{\mathrm{med}} & T_{\mathrm{air}} 
        \end{bmatrix}^\top.
        \label{eq:states_inputs_dis_outputs}
\end{split}
\end{equation} 
{$B_{\mathrm{dry}}$ is the dry biomass weight per larva, which is the variable to be observed.} The detailed dynamic description of each state is presented in the next subsections.
\subsection{Dry biomass weight $B_{\mathrm{dry}}$}
Dry weight of a larva changes due to the influences of the medium where the insects grow and the feed flux that is contained in the substrate.
This feed is assimilated and then converted into energy.
A part of this energy is for maintaing the larvae's structure and another part is for its growth and maturity.
This is mathematically described with a mass balance equation in \cite{padmanabha2020larvaemodel} using fluxes.
Omitting the maturity stage from the equation, the dynamic changes of the dry biomass weight per larva is described with
\begin{equation}
\begin{split}
    \dfrac{ \mathrm{d} B_{\mathrm{dry}}}{ \mathrm{d}t} =& \phi_{\mathrm{B_{eff}}} - \phi_{\mathrm{B_{maint}}} \\
    =& r_{\text {assim}}\left( B_{\text{dry}},T_{\text {med}} \right) k_{\text {inges}} B_{\text {dry}} \\ &-r_{\text {maint}}\left( T_{\text {med}} \right) k_{\text {maint}} B_{\text {dry}},
\end{split}
\end{equation}
where $r_{\text{assim}} \left( B_{\text{dry}}, T_{\text {med}} \right)$ and $r_{\text{maint}} \left( T_{\text {med}} \right)$ are the assimilation and maintenance rate functions, respectively. 
These rate functions depend on the {temperature and are} modelled as follows
\begin{equation}
\begin{split}
 &r_{\text{assim}} \left( B_{\text{dry}}, T_{\text {med}} \right) = \left( 1 - \frac{B_{\text{dry}}}{k_{\text {Basy }}} \right)\frac{r_{\text{T}} \left( T_{\text {med}} \right)}{k_{r_{\max }T},}\\
 & r_{\text{maint}} \left( T_{\text {med}} \right) = \frac{r_{\text{T}}\left( T_{\text {med}} \right)}{k_{r_{\max } T}},\\ 
        &r_{\mathrm{T}} \left( T_{\text {med}} \right) =k_{\mathrm{r}_{\text {max }} \mathrm{T}} \left[ 1+k_\gamma \exp \left(-k_{\rho \mathrm{T}}\left(T_{\text {med}}-k_{T_{\text {base}}}\right)\right) \right. \\ 
       &\left. +\exp \left(-\frac{k_{T_{\text {max}}}-T_{\text {med }}}{k_{\Delta T}}\right) \right]^{-1}.
\end{split}
\label{eq:rate_functions}
\end{equation}
$r_T\left( T_{\text {med}} \right)$ is the temperature {rate that describes} the growth of the insect with respect to the temperature $T_{\text {med}}$.
A modified Logan-10 model is used \cite{logan1976analytic}.
\begin{rem}
    Time dependencies are omitted to simplify the notations.
\end{rem}

\subsection{Temperature in the growing medium $T_{\mathrm{med}}$}
Temperature in the growing medium is caused by the metabolic activity of larvae and microbiome, conductive heat transfer to the walls and convective heat transfer from growing medium to air. 
This is mathematically described as
\begin{equation}
k_{\mathrm{C}_{\text {med }}} \frac{\mathrm{d} T_{\text {med }}}{\mathrm{d} t}=\phi_{\mathrm{Q}_{\text {bio }}}-\phi_{\mathrm{Q}_{\mathrm{m}-\mathrm{a}}}-\phi_{\mathrm{Q}_{\mathrm{m}-\mathrm{c}}},
\end{equation}
where the the fluxes are expressed by
\begin{equation}
    \begin{split}
        &\phi_{\mathrm{Q}_{\text {bio }}}=L_{\text {num }}\left(k_{\mathrm{Q}_{\text {assim }}} \phi_{\mathrm{B}_{\text {assim }}}+k_{\mathrm{Q}_{\text {maint }}} \phi_{\mathrm{B}_{\text {maint }}}+k_{\mathrm{Q}_{\text {mat }}} \phi_{\mathrm{B}_{\text {mat }}}\right),\\
        &\phi_{\mathrm{Q}_{\mathrm{a}-\mathrm{m}}}=k_{\mathrm{A}_{\mathrm{m}}} k_{\mathrm{h}_{\mathrm{a}-\mathrm{m}}}\left(T_{\text {med }}-T_{\mathrm{air}}\right),\\
        &\phi_{\mathrm{Q}_{\mathrm{m}-c}}=k_{\mathrm{A}_{\mathrm{m}-c}} k_{\mathrm{U}_{\mathrm{m}-c}}\left(T_{\mathrm{med}}-T_{\mathrm{chm}}\right),\\ 
    \end{split}
\end{equation}
and the heat capacities presented in \cite{PADMANABHA2023} are considered {constant, within a good approximation}
\begin{equation}
\begin{split}
    k_{\mathrm{h}_{\mathrm{a}-\mathrm{m}}} &=k_{\mathrm{he}_{\mathrm{a}-\mathrm{m}}} + k_{\mathrm{hm_{a-m}}}\phi_{\mathrm{W_{evap}}} \approx k_{\mathrm{he}_{\mathrm{a}-\mathrm{m}}},\\
     k_{\mathrm{C_{med}}} &= k_{\mathrm{c}_{\mathrm{chm}}} k_{\mathrm{m}_{\mathrm{chm}}}+k_{\mathrm{c}_{\text {water }}} W_{\text {cham }} \approx k_{\mathrm{m}_{\mathrm{chm}}}.
\end{split}  
\end{equation}
$T_{\mathrm{chm}}$, the temperature of the larvae's chamber, is assumed linearly dependent on the temperatures of the air and of the outside
\begin{equation}
    T_{\mathrm{chm}} \approx \frac{2}{3}T_{\mathrm{air}} + \frac{1}{3} T_{\mathrm{out}}. 
\end{equation}
This simplification reduces the number of states to be considered without heavily impacting the accuracy of the model.

\subsection{Temperature of the air $T_{\mathrm{air}}$}
The air temperature fluctuates due to various factors, including convective heat exchange between the air and the growing medium, heat fluxes generated by actuators such as heat coolers and ventilator pumps, convective flux between the air and the walls of the production environment, and heat losses resulting from leakage.
Heat losses associated with human interactions, such as opening and closing doors, are negligible and can be disregarded.
This is described as follows
\begin{equation}
k_{\mathrm{C}_{\mathrm{air}}} \frac{\mathrm{d} T_{\mathrm{air}}}{\mathrm{d} t}=\phi_{\mathrm{Q}_{\mathrm{hx}-\mathrm{a}}}+\phi_{\mathrm{Q}_{\mathrm{exch}}}+\phi_{\mathrm{Q}_{\mathrm{leak}}}+\phi_{\mathrm{Q}_{\mathrm{m}-\mathrm{a}}}\\+\phi_{\mathrm{Q}_{\mathrm{a}-\mathrm{c}}} 
\end{equation}
where fluxes are given by
\begin{equation}
\begin{split}
    \phi_{\mathrm{Q}_{\mathrm{a}-\mathrm{c}}} &= k_{\mathrm{A}_{\mathrm{c}}} k_{\mathrm{h}_{\mathrm{a}-c}}\left(T_{\mathrm{chm}}-T_{\mathrm{air}}\right),\\
    \phi_{Q_{a-m}} &= k_{\mathrm{A}_m} k_{\mathrm{h}_{a-m}}\left(T_{\text {med }}-T_{\text {air }}\right),\\
    \phi_{\mathrm{Q}_{\mathrm{a}-\mathrm{hx}}} &= k_{\mathrm{A}_{\mathrm{hx}}} k_{\mathrm{h}_{\mathrm{a}-\mathrm{hx}}}\left(T_{\mathrm{hx}}-T_{\mathrm{air}}\right),\\   \phi_{\mathrm{Qexch}} &= k_{\mathrm{c}_{\mathrm{air}}} k_{\rho_{\mathrm{air}}} \overbrace{k_{\dot{\mathrm{V}}_{\mathrm{u}}} u_{\mathrm{V}}}^{\phi_{\dot{\mathrm{v}}_{\mathrm{u}}}}\left(T_{\text {out }}-T_{\mathrm{air}}\right),\\
    \phi_{\mathrm{Q}_{\text {leak }}} &= k_{\mathrm{C}_{\text {air }}} k_{\rho_{\text {air }}} k_{\dot{\mathrm{v}}_{\text {leak }}}\left(T_{\text {out }}-T_{\text {air }}\right).   
\end{split}
\end{equation}
Some assumptions on the temperature of the air change are undertaken to simplify further.

First, total heat capacity of the air is constant
\begin{equation}
\begin{split}
    k_{\mathrm{C}_{\text {air }}}&=k_{\mathrm{c}_{\mathrm{air}}} k_{\rho_{\mathrm{air}}} k_{\mathrm{v}_{\mathrm{chm}}}+k_{\mathrm{C}_{\mathrm{vap}}} k_{\mathrm{v}_{\mathrm{chm}}} H_{\mathrm{air}} \\&\approx k_{\mathrm{c}_{\mathrm{air}}} k_{\rho_{\mathrm{air}}} k_{\mathrm{v}_{\mathrm{chm}}}.     
\end{split}
\end{equation}
Second, the temperature of the heat exchanger is considered linear with respect to the input signal to heater
\begin{equation}
    T_{\mathrm{hx}} \approx k_{\mathrm{hx}} u_{\mathrm{T}}.
\end{equation}

\subsection{{\emph{Hermetia illucens} temperature model: synthesis}}
The simplified three dimensional  state space model, using the states, inputs and outputs presented in \eqref{eq:states_inputs_dis_outputs}, is thus given by
\begin{equation}
\begin{split}
\dot{x}_{1} =&\left(k_{1}\left(1-\frac{1}{k_{2}}x_{1}\right)-k_{3}\right)r_{T}(x_2) x_{1},\\
\dot{x}_{2}  =&- k_{4}x_{2}+k_{5}x_{3}+k_{6}\left(1-\frac{1}{k_{2}}x_{1}\right)r_{T}(x_2)x_{1}\\ &+ k_{7}r_{T}(x_2)x_{1}- k_{8}d,\\
\dot{x}_{3} =&k_{9}x_{2}- k_{10}x_{3}+k_{11}u_{2}-k_{12}x_{3}u_{1}+k_{13}d\\ &+k_{12}du_{1},
\end{split}
\label{eq:state_space}
\end{equation}
and the outputs are
\begin{equation}
\begin{split}
y_{1} & =x_{2},\\
y_{2} & =x_{3}.
\end{split}
\end{equation}
All parameters are assumed to be known, and all inputs and outputs can be measured. 
For a detailed description of each parameter and their indices the reader is referred to \cite{PADMANABHA2023}.

\begin{table}[h!]
\vspace{6pt}
    \centering
    \begin{tabular}{ ||c| c|| }
    \hline
    Symbol & Expression \\
    \hline
    \hline 
    & \\
        $k_1$ & $k_{\mathrm{inges}}$\\
        $k_2$ & $k_{\mathrm{B_{asy}}}$ \\
        $k_3$ & $k_{\mathrm{maint}}$ \\
       $k_4$ & $\frac{1}{k_{\mathrm{C_{med}}}} \left( k_{\mathrm{A_m}}k_{\mathrm{h_{a-m}}} + k_{\mathrm{A_{m-c}}} k_{\mathrm{U_{m-c}}} \right) $ \\
        $k_5$ & $\frac{1}{k_{\mathrm{C_{med}}}} \left( k_{\mathrm{A_m}}k_{\mathrm{h_{a-m}}} + \frac{2}{3}k_{\mathrm{A_{m-c}}} k_{\mathrm{U_{m-c}}} \right)  $ \\
        $k_6$ & $\frac{1}{k_{\mathrm{C_{med}}}} L_{\mathrm{num}} k_{\mathrm{Q_{assim}}}k_{\mathrm{\alpha_{assim}}} k_{\mathrm{inges}}$ \\
        $k_{7}$ & $\frac{1}{k_{\mathrm{C_{med}}}} L_{\mathrm{num}} k_{\mathrm{Q_{maint}}}  k_{\mathrm{maint}}$ \\
        $k_{8}$ & $\frac{1}{k_{\mathrm{C_{med}}}} k_{\mathrm{A_{m-c}}} k_{\mathrm{U_{m-c}}}$ \\
         $k_9$ & $\frac{1}{k_{\mathrm{C_{air}}}} k_{\mathrm{A_{m}}} k_{\mathrm{h_{a-m}}}$ \\
          $k_{10}$ & $\frac{1}{k_{\mathrm{C_{air}}}} ( k_{\mathrm{A_{hx}}} k_{\mathrm{ha-hx}} + k_{\mathrm{C_{air}}} k_{\mathrm{\rho_{air}}} k_{\dot{\mathrm{V}}_{\mathrm{u}}} + \frac{1}{3}k_{\mathrm{A_c}}k_{\mathrm{h_{a-c}}}$ \\
           & $+ k_{\mathrm{A_m}} k_{\mathrm{h_{a-m}}} )$ \\
          $k_{11}$ & $\frac{1}{k_{\mathrm{C_{air}}}} k_{\mathrm{A_{hx}}} k_{\mathrm{ha-hx}} k_{\mathrm{hx}}$ \\
          $k_{12}$ & $k_{\mathrm{\rho_{air}}} k _{\dot{\mathrm{V}}_{\mathrm{u}}}$ \\
          $k_{13}$ & $ k_{\mathrm{\rho_{air}}} k _{\dot{\mathrm{V}}_{\mathrm{u}}} + \frac{1}{k_{\mathrm{C_{air}}}} k_{\mathrm{A_c}}k_{\mathrm{h_{ac}}}$\\
          & \\
          \hline
    \end{tabular}
    \caption{List of parameters expression. A more detailed list of all parameters, their descriptions and values can be found in \cite{PADMANABHA2023}.}
    \label{tab:parameters}
\end{table}
The obtained simplified \emph{Hermetia illucens} three dimensional temperature model was compared to the experimental results from the setup presented in \cite{PADMANABHA2023}. 
The weight of the dry biomass per larva, temperatures of air and medium were simulated using the same parameters, inputs and disturbance trajectories.
Other environmental factors, such as humidity, oxygen concentration, and substrate moisture, are implicitly accounted for in the parameters and are set to their optimal values based on the findings of \cite{PADMANABHA2020MPC}.
Before reaching maturity, the dry biomass weight per larva resulting from this simplified model follows the initial trajectory of the more complex model from which it is derived. This results in a smoother trajectory, albeit within the same range. It reaches a peak value and then ceases to evolve further.

Hence, the obtained model is only valid during the larvae growth process before the biomass weight's reaches its peak and the insects development process switches to maturity.
This is not problematic since the context of the modelling is in a farming process where the larvae is going to be harvested when its size is at its maximum.
The simulated air and medium temperature trajectories are also similar to the experimental data as they reach the same steady state.
The transitional phase of the temperatures, however, is different between simulation and experimental data. As numerous environmental factors are assumed to remain constant, several nonlinear behaviors are absent in the simplified model. Consequently, the trajectories are smoothed, and the simulations fail to capture the full spectrum of real temperature fluctuations.
But since this phase only lasts a few hours while the biomass evolves in a slower time scale that lasts weeks, it does not significantly impact the model.
Thus, the three dimensional \emph{Hermetia illucens} temperature model is an adequate candidate for analysis.

The observability of the model written in (\ref{eq:state_space}) is analysed in the next section.


\section{Observability analysis}
\label{Sec:Observability}
In this section, the observability analysis of the model presented in \eqref{eq:state_space} is addressed. 
The temperatures of the medium and air can be measured, therefore $x_2$ and $x_3$ are known. 
Similarly, the temperature of the outside is also given. It is considered as a disturbance because in practice this temperature cannot be influenced or controlled.
Since in practice biomass weight is often the most difficult state to measure in real time, we are interested in the possibility of estimating the dry biomass weight of the larvae, i.e. $x_1$. 
For this we consider the system in \eqref{eq:state_space} and study its observability. Since the states $x_2$ and $x_3$ are measured, the task is reduced to determine if $x_1$ can be calculated from the measurements.
The following definitions are needed to explain the observability concepts.
\begin{dfn}[Analytic function \cite{ComplexAnalysis}]
A function $h$ is said to be an analytic function on an open set $\mathcal{X}$ if for any point $x_0 \in \mathcal{X}$:
\begin{itemize}
\item $h$ is infinitely differentiable on $\mathcal{X}$,
\item $h$ is locally given by a convergent power series.
\end{itemize}
\end{dfn}

\begin{dfn}[Instantaneous observability \cite{BERNARD2022}]
 A system is instantaneously observable if for all  $t_d > 0$ and for $u_0 \in \mathcal{U}$, any pair of solutions     $x_a(t)$                   and $x_b(t)$
initialized in  $ \mathcal{X} $  and defined on $\mathbb{R}_{+}$               such that 
    \begin{equation*}
    h(x_{a}(t),u_0) = h(x_{b}(t), u_0), \; \; \forall 0 \leq t < t_{d}
\end{equation*} 

verifies  $x_a(\cdot) = x_b(\cdot)$.
\end{dfn}

\begin{dfn}[Differential observability \cite{Bernard2019}]
 A system is differentially observable if there exists $n_{y} \in \mathbb{N}$ such that the mapping $H : \mathcal{X} \ra \mathbb{R} ^ {n_{y}}$,
 \begin{equation}
\begin{aligned}
&H(x)=\left(\begin{array}{c}
y \\
\dot{y} \\
\vdots \\
\frac{d^{(n_y-1)}y}{dt^{(n_y-1)}}
\end{array}\right)
\end{aligned}
\label{eq:map_H}
\end{equation}
is injective and full rank on $\mathcal{X}$. 
\end{dfn}

\begin{asm}
   All output functions considered in this paper are analytic functions.
\end{asm}
Based on this assumption, observability is analyzed in the sense of the above definitions.

The textbook observability analysis method for nonlinear systems was introduced in \cite{hermann1977nonlinear}, and is commonly known as checking the rank condition. 
This approach permits to check the local differential observability of the system  by investigating the local injectivity of the map $H$, that is presented in (\ref{eq:map_H}).
This map is obtained by taking successive Lie derivatives of the system's output function, and it depends on the inputs and states. 
If its Jacobian matrix is full column rank, then it is locally injective and the local observability is assured.
However, we are interested in obtaining a global result of differential observability.

We study the \emph{differential observability} of model \eqref{eq:state_space}. This classical concept (see e.g. \cite{Ber19}, \cite{BERNARD2022}, \cite{GauKup01}) requires, for our problem at hand, that for some $m\geq 1$ there exists an \emph{injective} mapping $\Omega: \mathbb{R}_{\geq0} \rightarrow \mathbb{R}^{m}$ between the state $x_{1} \in \mathbb{R}_{\geq0}$ and the measured signals $\left( y,u,d \right)$ and a finite number of their derivatives, i.e. 
\[
\Omega(x_1) = \psi\left(y,\cdots,y^{(m-1)},u,\cdots,u^{(m-1)},d,\cdots,d^{(m-1)} \right)
\]
for a known function $\psi(\cdot) \in \mathbb{R}^{m}$. Injectivity of $\Omega$ means that if at any time instant we know $\left( y,u,d \right)$ and their derivatives then the value of $x_1$ is uniquely determined.

The main result of the section is the next theorem that establishes conditions for the global differential observability of the model \eqref{eq:state_space} depending only on the values of the parameters.

\begin{thm} 
Assume that all parameters of the model \eqref{eq:state_space} are positive. If the following condition 
\begin{equation}
  k_{7}  \geq \left( 1 - 2\frac{k_{3}}{k_{1}} \right) k_{6} 
  \label{eq:cond_injectivity}
\end{equation}
is satisfied then the system is globally and differentially observable for all $x_1 \in \mathbb{R}$.
\end{thm}

\begin{proof}
The procedure to find the map $\Omega$ is to take successive derivatives of the output variables $y$. 
Then, candidate maps $\Omega$ are constructed from these derivatives. 
Each candidate will be checked for injectivity until an adequate injective map is constructed.

From the dynamics \eqref{eq:state_space} it is immediate that the output $y_2$ and its first derivative are not related to $x_1$, and so we only consider derivatives of the output $y_1$.

The first step is to analyse the first derivative of $y_1=x_2$, which is readily obtained from \eqref{eq:state_space} and it is given by
\begin{equation}
    \begin{split}
\dot{y}_{1}=&-k_{4}y_{1}+k_{5}y_{2}+\left[k_{6}\left(1-\frac{1}{k_{2}}x_{1}\right)+k_{7}\right]r_{T}\left(y_{1}\right)x_{1}\\& - k_{8}d.    
    \end{split}
    \label{eq:dot_y}
\end{equation}
Taking into consideration that the temperature rate function \eqref{eq:rate_functions} is by definition strictly positive, i.e. $r_{\mathrm{T}}(y_1) > 0$ for all $y_1$, we shift the measurements and disturbance on the right hand side and the $x_1$ depending function on the left hand side.
Therefore, this expression can be rewritten as
\begin{equation}
\Omega_{1}\left(x_1\right) = \frac{\dot{y}_{1}+k_{4}y_{1}-k_{5}y_{2} + k_{8}d }{r_{T}\left(y_{1}\right)},
\label{eq:Obs_phi}
\end{equation}
where
\begin{equation}
\Omega_{1}\left(x_{1}\right) \triangleq \left[k_{6}\left(1-\frac{1}{k_{2}}x_{1}\right)+k_{7}\right]x_{1}.\label{eq:Def_phi}
\end{equation}
Note that the right hand side of equation \eqref{eq:Obs_phi} depends only on measured variables and their derivatives, and it can therefore be calculated from the measurements. 

Injectivity of function $\Omega_1$ would imply instantaneous observability. However, the graph of $\Omega_1$ is a parabola with two real roots, at $x_1=0$ and at $x_1 = k_2\left(1+\frac{k_{7}}{k_{6}} \right)$, and it attains its maximum at the point (see Fig. \ref{fig:phi_x1})
\[
x_{1}^{*}=\frac{k_{6}+k_{7}}{2k_{6}}k_{2}.  
\]
It is clear that, if the domain of $x_1$ is not restricted to the intervals $0\leq x_1 \leq x_{1}^{*}$ or $x_1 \geq x_{1}^{*}$, then $\Omega_1$ is not injective. Checking this condition with the parameters reported in \cite{PADMANABHA2023}, it is apparent that $x_1$ is not confined in any of the intervals.
\begin{figure}
\vspace{6pt}
 \centering
    \includegraphics[width=0.75\linewidth]{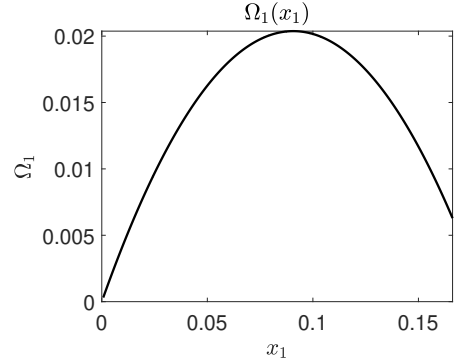}
    \caption{Evolution of function $\Omega_1$ with respect to $x_1$}
    \label{fig:phi_x1}
\end{figure}
\begin{rem}
\label{rem:inject} We note, for future use, that variable $x_1$ is restricted to be positive, i.e. $x_1 \geq 0$. For $x_1$ in the interval $0 \leq x_1 \leq k_2\left(1+\frac{k_{7}}{k_{6}} \right)$ there is no injectivity and the value of $\Omega_1$ satisfies 
\[
0 \leq \Omega_1(x_1) \leq \Omega_{1}(x_1^*) = \frac{(k_{6} + 3k_{7})(k_{6} + k_{7})}{4k_{6}}k_2 .
\] 
However, when $x_1 > k_2\left(1+\frac{k_{7}}{k_{6}} \right)$ there is injectivity, since only one root of the parabola $\Omega_1$ is positive.
\end{rem}

The next step is to investigate the second derivative of $y_1$, from wich we will construct another map.

Differentiating \eqref{eq:dot_y} we obtain
\begin{equation}
\begin{split}
    \Ddot{y}_1 =& - k_4 \dot{y}_1 + k_5 \dot{y}_2 + \Omega_{1}^{'}(x_1) \dot{x}_1 r_T(y_1) \\ &+ \Omega_{1}(x_1)\dot{r}_T(y_1) + k_8 \dot{d}.   
\end{split}
\end{equation}
Use of \eqref{eq:Obs_phi} and \eqref{eq:Def_phi} leads to
\begin{equation}
\begin{split}
&\Ddot{y}_1 = - k_4 \dot{y}_1 + k_5 \dot{y}_2\\
     &+ \left[ \left(k_{6}+k_{7}-2\frac{k_{6}}{k_{2}}x_{1} \right)\left(k_{1}  -k_{3} -\frac{k_{1}}{k_{2}}x_{1} \right)x_{1}\right] r^2_{T}(y_1)\\
     & + \frac{\dot{y}_{1}+k_{4}y_{1}-k_{5}y_{2}-k_{8}d}{r_{T}(y_1)} \dot{r}_T(y_1,\dot{y_1}) + k_{8} \dot{d}.       
\end{split}
\end{equation}
Similarly to the construction of $\Omega_1(x_1)$, we shift the measurements and disturbance on the right hand side while keeping the $x_1$ depending function on the left hand side.
This expression can be rewritten as
\begin{equation}
\begin{split}
    \Omega_{2}\left( x_1 \right) &=  \frac{\Ddot{y}_1 + k_4 \dot{y}_1 -  k_5 \dot{y}_2 + k_8 \dot{d}}{r^2_{T}(y_1)} \\
      &  - \frac{\dot{y}_{1}+k_{4}y_{1}-k_{5}y_{2}+k_{8}d}{r_{T}^{3}(y_1)}\dot{r}_T(y_1,\dot{y_1}),  
\end{split}
\label{eq:T2}
\end{equation}
where
\begin{equation}
  \Omega_{2}\left( x_1 \right) \triangleq \left(k_{6}+k_{7}-2\frac{k_{6}}{k_{2}}x_{1} \right)\left(k_{1}  -k_{3} -\frac{k_{1}}{k_{2}}x_{1} \right)x_{1}.
  \label{eq:xi_x}
\end{equation}

Note that the right hand side of \eqref{eq:T2} depends only on measured variables and their derivatives, so that it can be calculated from the measurements.

$\Omega_{2}\left( x_1 \right)$ in \eqref{eq:xi_x} is a cubic function, with three real roots, so that it is clearly not injective (see Fig. \ref{fig:xi_x1}), and $\Omega_2$ alone cannot assure instantaneous observability. 
\begin{figure}
 \centering
 \includegraphics[width=0.75\linewidth]{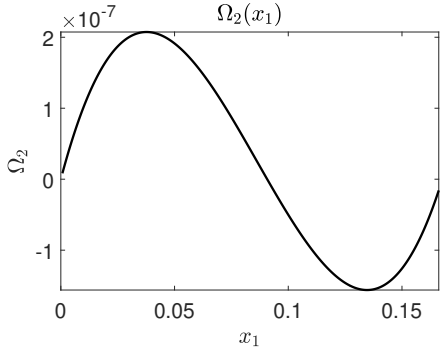}
    \caption{Evolution of function $\Omega_2$ with respect to $x_1$}
    \label{fig:xi_x1}
\end{figure}

Finally, we study the injectivity of the function 
\begin{equation}
\Omega:x_{1}\in \mathcal{X}\subset\mathbb{R}_{\geq 0}\rightarrow\left(\begin{array}{c}
\Omega_1 \left(x_{1}\right)\\
\Omega_2 \left(x_{1}\right)
\end{array}\right)\in\mathbb{R}^{2}.    
\end{equation}
The injectivity of $\Omega$ corresponds to the instantaneous observability of the dynamical model. Recall that $\Omega$ is injective iff 
\begin{equation}
   \left(\begin{array}{c}
\Omega_1 \left(x_{1}^{(1)}\right)\\
\Omega_2 \left(x_{1}^{(1)}\right)
\end{array}\right)=\left(\begin{array}{c}
\Omega_1 \left(x_{1}^{(2)}\right)\\
\Omega_2 \left(x_{1}^{(2)}\right)
\end{array}\right)\Rightarrow x_{1}^{(1)}=x_{1}^{(2)}.
\end{equation}
According to Remark \ref{rem:inject} when $x_1 > k_2 (1+\frac{k_{7}}{k_{6}})$ function $\Omega_1$ is injective (in the domain $x_1 \geq 0$), and thus we can restrict the analysis of injectivity of $\Omega$ to the interval $0\leq x_1 \leq k_2 (1+\frac{k_{7}}{k_{6}})$. 

To find when the mapping $\Omega$ is injective we determine by contraposition when it is not injective. For this we find when $x_{1}^{(1)} \neq x_{1}^{(2)} \Rightarrow \Omega(x_{1}^{(1)}) = \Omega(x_{1}^{(2)})$. We just need to consider a pair $x_{1}^{(1)} \neq x_{1}^{(2)}$ such that $\Omega_{1}(x_{1}^{(1)}) = \Omega_{1}(x_{1}^{(2)})$ and then determine when  $\Omega_{2}(x_{1}^{(1)}) = \Omega_{2}(x_{1}^{(2)})$.

Values $x_{1}^{(1)} \neq x_{1}^{(2)}$ for which $\Omega_{1}(x_{1}^{(1)}) = \Omega_{1}(x_{1}^{(2)}) = \nu$ are easily found, since $\Omega_{1}$ is a quadratic function, and they are given by
\[
x_{1}^{\left(1,2\right)}= \frac{k_{2}}{2k_{6}}\left(\left(k_{6}+k_{7}\right)\pm\sqrt{\Delta}\right),
\]
where $\Delta \triangleq \left(k_{6}+k_{7}\right)^{2}-4\frac{k_{6}}{k_{2}}\nu$. Note that $x_{1}^{(1)} \neq x_{1}^{(2)}$ and are real numbers only if $\nu < \frac{(k_{6} + k_{7} )^{2}k_2}{4 k_{6}}$, i.e. $\Delta >0$.  

Now we find out when $\Omega_{2}(x_{1}^{(1)}) = \Omega_{2}(x_{1}^{(2)})$. Using \eqref{eq:xi_x} and after some algebraic simplifications we conclude that this is the case if
\begin{multline}
  -\left(k_{1}  -k_{3} -\frac{k_{1}}{2k_{6}}\left(k_{67} + \sqrt{\Delta}\right) \right)\left(k_{67} + \sqrt{\Delta}\right) = \\
  \left(k_{1}  -k_{3} - \frac{k_{1}}{2k_{6}}\left(k_{67} - \sqrt{\Delta}\right) \right) \left(k_{67} - \sqrt{\Delta}\right),        
\end{multline}
where we have used that $\Delta>0$ and $k_{67} = k_{6} + k_{7}$. This expression can be easily further simplified to
\[
\frac{k_{2}}{2k_{6}}\left(k_{7}+\frac{k_{6}}{k_{1}}k_{3}\right)\left(k_{6}+k_{7}\right)= \nu .
\]
Since $\nu < \frac{(k_{6} + k_{7} )^{2}k_2}{4 k_{6}}$ the latter equation is feasible if
\[
\frac{k_{2}}{2k_{6}}\left(k_{7}+\frac{k_{6}}{k_{1}}k_{3}\right)\left(k_{6}+k_{7}\right) < \frac{(k_{6} + k_{7} )^{2}k_2}{4 k_{6}} ,
\]
i.e.
\[
  k_{7}  < \left( 1 - 2\frac{k_{3}}{k_{1}} \right) k_{6} .
\]
This is the condition for $\Omega$ to be not injective. Thus 
\[
  k_{7}  \geq \left( 1 - 2\frac{k_{3}}{k_{1}} \right) k_{6} 
\]
implies injectivity.

$\Omega$ is injective as long as \eqref{eq:cond_injectivity} is satisfied. 
Since we have previously established that the biomass can be instantaneously observed when $\Omega$ is injective on $\mathcal{X}$, then the theorem is proved. 
\end{proof}

\begin{rem}
It is also possible to verify the injectivity of $\Omega$ is by plotting the graph of the function $\Omega$ on the
plane $\left(\begin{array}{c}
\Omega_1 \left(x_{1}\right)\\
\Omega_2 \left(x_{1}\right)
\end{array}\right)$ for all values of $x_{1}\in\left[ 0, k_2 (1+\frac{k_{7}}{k_{6}}) \right]$. 
This graph is a curve in $\mathbb{R}^{2}$ (see Fig. \ref{fig:Omega}). The curve intersects itself if and only if the function is not injective. 
\begin{figure}
 \centering
    \includegraphics[width=0.75\linewidth]{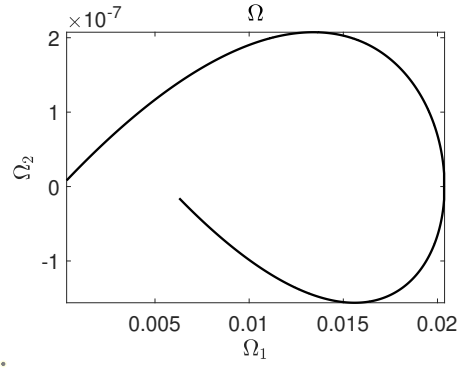}
    \caption{Behaviour of the curve $\Omega(x_1)$}
    \label{fig:Omega}
\end{figure}
This approach is very simple and appealing. However, the issue with such an approach is that it relies on the available system's parameters. 
This means that the result is only valid for one unique scenario. 
If the parameters change, then this result cannot be exploited anymore.
For the general case, the analytical method is preferred.  
\end{rem}

Since the existence of a global injective map between the state space and inputs/outputs space was proven to exist, then the three dimensional model of the \emph{Hermetia illucens} is globally instantaneously observable.
This result gives a theoretical guarantee of the possibility of estimating $x_1$, the dry biomass weight per larva, in any arbitrary interval of time. 

Moreover, since the observability map requires two derivatives of the output to achieve injectivity, then a classical observer, being a copy of the system and a static output injection term is not appropriate. An immersion technique will be more akin to the observability properties of the system. The design of an appropriate observer is under study and it will be reported elsewhere.

\section{Conclusion}
\label{Sec:Conclusion}
A simplification of an extensive dynamical model of the growth process of the \emph{Hermetia illucens} insects focusing on the dry biomass per larva and the different temperatures that impacts its behaviour was presented. 
Some assumptions were undertaken to obtain the resulted model but the same parameters and experience scenarios found in the literature were kept. 
This was done in order to analyze the global observability of the biomass of such a model.
As expected, the observability analysis of the \emph{Hermetia illucens} model depends on the dry biomass weight only. 
Indeed, in practice the temperatures can be measured in real time contrary to the biomass weight.
A differential observability approach was exploited to derive a general observability condition on the parameters of the model. 
Under this condition, the proof of the global instantaneous observability of the dry biomass was presented.
This gives the guarantee of the convergence of both asymptotic and nonasymptotic (or finite time) observers based on the previous model. 
Future work will focus on the design of such an observer.

\addtolength{\textheight}{-3cm}   





\phantomsection
\addcontentsline{toc}{chapter}{Bibliography} 
\bibliographystyle{acm}
\bibliography{bib}

\end{document}

%% file: ACSD-preamble.tex
\usepackage{xspace}
\usepackage{soul}
\usepackage{mathtools}
\usepackage{makecell}
\usepackage{cite}
\usepackage{nicefrac}
\usepackage{float}
\usepackage{wrapfig}
\usepackage{color}

\newcommand{\ra}{\rightarrow}

\newtheorem{dfn}{Definition}
\newtheorem{asm}{Assumption}
\newtheorem{thm}{Theorem}

\newtheorem{rem}{Remark}

\makeatletter
\newcommand{\subalign}[1]{%
	\vcenter{%
		\Let@ \restore@math@cr \default@tag
		\baselineskip\fontdimen10 \scriptfont\tw@
		\advance\baselineskip\fontdimen12 \scriptfont\tw@
		\lineskip\thr@@\fontdimen8 \scriptfont\thr@@
		\lineskiplimit\lineskip
		\ialign{\hfil$\m@th\scriptstyle##$&$\m@th\scriptstyle{}##$\crcr
			#1\crcr
		}%
	}
}
